\documentclass[a4paper]{article}

\usepackage[T1]{fontenc}
\usepackage[utf8]{inputenc}
\usepackage[english]{babel}
\usepackage{graphics}
\usepackage{epsfig}
\usepackage{amsmath}
\usepackage{amssymb}
\usepackage{amsthm}
\usepackage{url}
\usepackage[breaklinks]{hyperref}
\usepackage{float}
\usepackage{pifont}
\usepackage[table]{xcolor}
\usepackage{lmodern}
\usepackage{xspace}
\usepackage{bbold}
\usepackage[capitalize,nameinlink]{cleveref}
\usepackage{textcomp}

\crefname{section}{Section}{Sections}
\crefname{subsection}{Section}{Sections}
\crefname{definition}{Definition}{Definitions}

\usepackage{draftwatermark} 
\SetWatermarkText{Preprint} 
\SetWatermarkScale{1}       

\DeclareMathOperator{\R}{\mathbb{R}}
\DeclareMathOperator{\1}{\mathbb{1}}

\DeclareMathOperator{\e}{e}
\DeclareMathOperator{\D}{\mathrm{d}\!}
\DeclareMathOperator{\lnorm}{\mathcal{L}}

\DeclareMathOperator{\tr}{tr}

\DeclareMathOperator{\SVD}{SVD}
\DeclareMathOperator{\EVD}{EVD}

\newcommand{\T}{\intercal}

\newcommand{\emgr}{\texttt{emgr}\xspace }

\usepackage[normalem]{ulem} 
\newcommand{\hl}[1]{\bgroup\markoverwith
  {\textcolor{#1}{\rule[-.5ex]{2pt}{2.5ex}}}\ULon}

\title{\emgr\ - The Empirical Gramian Framework}
\author{Christian Himpe\thanks{ORCiD: \href{http://orcid.org/0000-0003-2194-6754}{0000-0003-2194-6754}, Contact: \href{mailto:himpe@mpi-magdeburg.mpg.de}{\nolinkurl{himpe@mpi-magdeburg.mpg.de}},  
Computational Methods in Systems and Control Theory Group at the Max Planck Institute for Dynamics of Complex Technical Systems, Sandtorstra{\ss}e~1, D-39106 Magdeburg, Germany}}

\date{}

\newtheoremstyle{thm}{\topsep}{\topsep}{\normalfont \itshape}{}{\normalfont \bfseries}{}{\newline}{}
\theoremstyle{thm}

\newcounter{dummy}

\newtheorem{mycorollary}{Corollary}
\newtheorem{mydefine}[dummy]{Definition}

\begin{document}

\maketitle

\begin{abstract}\bfseries\noindent 
System Gramian matrices are a well-known encoding for properties of input-output systems such as controllability, observability or minimality. 
These so-called system Gramians were developed in linear system theory for applications such as model order reduction of control systems.
Empirical Gramian are an extension to the system Gramians for parametric and nonlinear systems as well as a data-driven method of computation.
The empirical Gramian framework - \emgr ~- implements the empirical Gramians in a uniform and configurable manner,
with applications such as Gramian-based (nonlinear) model reduction, decentralized control, sensitivity analysis, parameter identification and combined state and parameter reduction.
\end{abstract}

\begin{tabular}{rl}
 \textbf{Keywords:} & Model Reduction; Model Order Reduction; Decentralized Control; \\
                    & Sensitivity Analysis; Parameter Identification; Empirical Gramians; \\
                    & Nonlinear Systems; Reduced Order Systems; Controllability; Observability
\end{tabular}

~\\

\begin{tabular}{rl}
 \textbf{PACS:} & 02.30.Yy \\
 \textbf{MSC:} & 93A15, 93B20, 93C10 \\
\end{tabular}






\section*{Code Meta Data} 

\begin{center}
\begin{tabular}{rl}
 name (shortname) & EMpirical GRamian Framework (\texttt{emgr}) \\
 version (release-date) & 5.4 (2018-05-05) \\
 identifier (type) & \href{http://doi.org/10.5281/zenodo.1241532}{\texttt{doi:10.5281/zenodo.1241532}} (doi) \\
 authors (ORCIDs) & Christian Himpe (0000-0003-2194-6754) \\
 topic (type) & Model Reduction (toolbox) \\
 license (type) & 2-Clause BSD (open) \\
 repository (type) & \href{http://github.com/gramian/emgr}{\texttt{git:github.com/gramian/emgr}} (git) \\
 languages & Matlab \\
 dependencies & OCTAVE >= 4.2, MATLAB >= 2016b \\
 systems & Linux, Windows \\
 website & \url{http://gramian.de} \\
 keywords & empirical gramians, cross gramian, combined reduction
\end{tabular}
\end{center}

\pagebreak[4]

\section{Introduction} 
Attributes of input-output systems, such as controllability, observability or minimality can be assessed by special matrices.
These so-called system Gramian matrices, or short system Gramians\footnote{Also: Grammians or Gram matrix.}, have manifold applications in system theory, control theory and mathematical engineering.

Originally, Gramian-based methods were developed for linear systems \cite{Kal63}.
The empirical Gramian matrices \cite{morLalMG99} are an extension of Gramian-based methods to nonlinear and parametric systems.
This work summarizes the \textbf{em}pirical \textbf{Gr}amian framework (\emgr) \cite{morHim18a}, a compact software toolbox,
which implements various types of empirical Gramians as well as the related empirical covariance matrices.

An important use of empirical Gramian matrices is model order reduction (MOR),
utilizing the capability of system Gramians to quantify the input-output importance of the underlying system's states based on controllability and observability.
Several variants of Gramian-based model reduction are available, for example:
\begin{itemize}
 \item Linear Model Order Reduction \cite{morMoo81},
 \item Robust Model Reduction \cite{morSunH06a},
 \item Parametric Model Order Reduction (pMOR) \cite{morHimO15a},
 \item Nonlinear Model Order Reduction (nMOR) \cite{morLalMG99,morHahE00,morConI04a,morYaoDY08},
 \item Second-Order System Model Reduction \cite{morZhaC08},
 \item Combined State and Parameter Reduction \cite{morHimO14}.
\end{itemize}
Beyond model reduction, (empirical) system Gramians can also be utilized for tasks like:
\begin{itemize}
 \item Sensitivity Analysis \cite{StrFB06,LysHA16},
 \item Parameter Identification \cite{GefFSetal08,morTolA17},
 \item Decentralized Control \cite{MoaK06,ShaK12,ShaS13},
 \item Optimal Sensor Placement \cite{SinH05,SalOWetal16}, Optimal Actuator Placement \cite{SumCL15},
 \item Optimal Control \cite{morLawMC05}, Model Predictive Control \cite{morHahKE02},
 \item Nonlinearity Quantification \cite{morHahE01,JiaWJetal16}.
\end{itemize}
Also, various system invariants and indices are computable using system Gramians, and thus also by empirical Gramians:
\begin{itemize}
 \item System gain \cite{StrFB06},
 \item Cauchy index \cite{morFerN83a,ForF12},
 \item Information entropy index \cite{morFuZDetal10},
 \item Nyquist plot enclosed area \cite{HalCB10},
 \item System Frobenius norm and ellipsoid volume \cite{HriHC14},
 \item System $H_2$-norm \cite{morGugAB08}.
\end{itemize}

\pagebreak

This wide range of applications and the compatibility to nonlinear systems make empirical Gramians a versatile tool in many system-theoretic computations.
Furthermore, the empirical Gramians provide a data-driven method of computation with close relations to proper orthogonal decomposition (POD) and balanced POD (bPOD) \cite{morWilP02,morRow05}.

Various (Matlab) implementations are available for the computation of linear system Gramians by the solution of associated matrix equations,
such as the basic \texttt{sylvester} command, the \texttt{gram}, \texttt{lyap} and \texttt{covar} commands from the \textsc{Matlab} Control Toolbox\footnote{\url{http://mathworks.com/products/control}} and \textsc{Octave} Control Package\footnote{\url{http://octave.sourceforge.net/control}}.
For empirical Gramians, the only other generic implementation, to the author's best knowledge, is \cite{morSunH06b},
which provides only the empirical controllability Gramian and the empirical observability Gramian,
but \emph{not} any empirical cross Gramian (see \cref{def:wy,def:wx}).
This makes \emgr the unique (open-source) implementation of all three: the empirical controllability Gramian $W_C$, the empirical observability Gramian $W_O$ and the empirical cross Gramian $W_X$ (sometimes also symbolized by $W_{CO}$ and $X_{CG}$).

Lastly, it is noted that the term \emph{empirical Gramian} is used as an umbrella term for the original empirical Gramians \cite{morLalMG99}, the empirical covariance matrices \cite{morHahE02}, modified empirical Gramians \cite{morChoSB16} or local Gramians \cite{KreI09}.

\subsection{Aim} 
After its initial version $1.0$ (2013) release, accompanied by \cite{morHimO13},
the empirical Gramian framework\footnote{\texttt{emgr} is also listed in the \textbf{Oberwolfach References on Mathematical Software} (ORMS), no.~345: \url{http://orms.mfo.de/project?id=345}.} has been significantly enhanced.
Apart from extended functionality and accelerated performance,
various new concepts and features were implemented.
Now, with the release of version $5.4$ (2018) \cite{morHim18a},
this is a follow-up work illustrating the current state of \emgr and its applicability,
 as well as documenting the flexibility of this toolbox.
In short, the major changes involve:
\begin{itemize}
 \item Non-symmetric cross Gramian variant,
 \item linear cross Gramian variant,
 \item distributed cross Gramian variant and interface,
 \item inner product kernel interface,
 \item time-integrator interface,
 \item time-varying system compatibility,
 \item tensor-based trajectory storage,
 \item functional paradigm software design.
\end{itemize}

\subsection{Outline} 
This work is structured as follows: In \cref{sec:prelim} the empirical Gramian's main application,
projection-based model order reduction, is briefly described;
followed by \cref{sec:eg} presenting the mathematical definitions of the computable empirical Gramians.
\cref{sec:impl} summarizes the design decision for \emgr, while \cref{sec:intf} documents usage and configuration.
Numerical examples are demonstrated in \cref{sec:numex} and lastly,
in \cref{sec:sum}, a short concluding remark is given.

\section{Mathematical Preliminaries}\label{sec:prelim} 
The mathematical objects of interest are nonlinear parametric input-output systems,
which frequently occur in physical, chemical, biological and technical models or spatially discretized partial differential equations (PDE).
These control system models consist of a dynamical system (typically on $\R$, i.e. an ordinary differential equation (ODE)) as well as an output function,
and maps the input $u : \R_{>0} \to \R^M$ via the state $x : \R_{>0} \to \R^N$ to the output $y : \R_{>0} \to \R^Q$:
\begin{align}\label{eq:iosys}
\begin{split}
 \dot{x}(t) &= f(t,x(t),u(t),\theta), \\
       y(t) &= g(t,x(t),u(t),\theta).
\end{split}
\end{align}
The potentially nonlinear vector-field $f:\R_{>0} \times \R^N \times \R^M \times \R^P \to \R^N$ and output functional \mbox{$g:\R_{>0} \times \R^N \times \R^M \times \R^P \to \R^Q$} both depend on the time $t \in \R_{>0}$, the state $x(t)$, input or control $u(t)$ and the parameters $\theta \in \R^P$.
Together with an initial condition $x(0) = x_0 \in \R^N$, this setup constitutes an initial value problem.

\subsection{Model Reduction} 
The aim of model reduction is the algorithmic computation of surrogate reduced order models with lower computational complexity or memory footprint than the original full order model.
For the sake of brevity only combined state and parameter reduction is summarized here,
which includes state-space reduction, parametric state-space reduction and parameter-space reduction as special cases;
for an elaborate layout see \cite{morHimO16}.

Given the general, possibly nonlinear, input-output system \eqref{eq:iosys}, a combined state and parameter reduced order model: 
\begin{align*}
 \dot{x}_r(t) &= f_r(t,x_r(t),u(t),\theta_r), \\
 \tilde{y}(t) &= g_r(t,x_r(t),u(t),\theta_r), \\
       x_r(0) &= x_{r,0},
\end{align*}
with a reduced state $x_r : \R_{>0} \to \R^n$, $n \ll N$, and a reduced parameter \mbox{$\theta_r \in \R^p$}, $p \ll P$, is sought.
Accordingly, a reduced vector-field \mbox{$f_r :\R_{>0} \times \R^n \times \R^M \times \R^p \to \R^n$} and a reduced output functional \mbox{$g_r:\R_{>0} \times \R^n \times \R^M \times \R^p \to \R^Q$} describe the reduced system,
for which the reduced system's outputs $\tilde{y} : \R_{>0} \to \R^Q$ should exhibit a small error compared to the full order model,
yet preserving the parameter dependency:
\begin{align*}
 \|y(\theta)-\tilde{y}(\theta_r)\| \ll 1.
\end{align*}
A class of methods to obtain such a reduced order model with the associated requirements is described next.

\subsubsection{Projection-Based Combined Reduction} 
Projection-based combined state and parameter reduction is based on (bi-) orthogonal truncated projections for the state- and parameter-space respectively.
The state-space trajectory $x(t)$, not too far from a steady-state $\bar{x} \in \R^N$, $\bar{u} \in \R^M$, $f(t,\bar{x},\bar{u},\theta) = 0 \;\; \forall t$,
is approximated affinely using truncated reducing and reconstructing projections $U_1 \in \R^{N \times n}$ and $V_1 \in \R^{N \times n}$, with $V_1^\T U_1 = \1_n$:
\begin{align*}
 x_r(t) := V_1^\T (x(t) - \bar{x}) \Rightarrow x(t) \approx \bar{x} + U_1 x_r(t).
\end{align*}
The relevant parameter-space volume is also approximated by truncated reducing and reconstructing projections $\Pi_1 \in \R^{P \times p}$ and $\Lambda_1 \in \R^{P \times p}$, with $\Lambda_1^\T \Pi_1 = \1_p$:
\begin{align*}
 \theta_r := \Lambda_1^\T (\theta - \bar{\theta}) \Rightarrow \theta \approx \bar{\theta} + \Pi_1 \theta_r,
\end{align*}
relative to a nominal parameter $\bar{\theta}$.
Given these truncated projections, a projection-based reduced order model is then obtained by:
\begin{align*}
 \dot{x}_r(t) &= V_1^\T f(t,\bar{x} + U_1 x_r(t),u(t),\bar{\theta} + \Pi_1 \theta_r), \\
 \tilde{y}(t) &= \phantom{V_1^\T} g(t,\bar{x} + U_1 x_r(t),u(t),\bar{\theta} + \Pi_1 \theta_r), \\
       x_r(0) &= V_1^\T (x_0 - \bar{x}), \\
  \theta_r    &= \Lambda_1^\T (\theta - \bar{\theta}).
\end{align*}
Thus, to obtain a projection-based reduced order model with respect to the state- and parameter-space,
the overall task is determining the truncated projections $U_1$, $V_1$, $\Lambda_1$ and $\Pi_1$.

It should be noted, that this approach produces globally reduced order models,
meaning $U_1$, $V_1$, $\Lambda_1$, $\Pi_1$ are valid over the whole operating region,
which is an application-specific subspace of the Cartesian product of the full order state- and parameter-space $\R^N \times \R^P$.

\subsubsection{Gramian-Based Model Reduction} 
Gramian-based model reduction approximates the input-output behavior of a system by removing the least controllable \emph{and} observable state components.
To this end the system is transformed to a representation in which controllability and observability are balanced.
Given a controllability Gramian $W_C$ (\cref{def:wc}) and observability Gramian $W_O$ (\cref{def:wo}) to an input-output system, a balancing transformation \cite{morMoo81} is computable;
here in the variant from \cite{morGarB02}, utilizing the singular value decomposition (SVD): 
\begin{align*}
 W_C &\stackrel{\SVD}{=} U_C D_C U_C^\T, \\
 W_O &\stackrel{\SVD}{=} U_O D_O U_O^\T \\
 \rightarrow U_C D_C^{\frac{1}{2}} U_C^\T U_O D_O^{\frac{1}{2}} U_O^\T &= W_C^{\frac{1}{2}} W_O^{\frac{1}{2}} \stackrel{\SVD}{=} U D V^\T.
\end{align*}
Partitioning the columns of $U$ and $V$ based on the (Hankel) singular values in $D$, $D_{ii} = \sigma_i < \sigma_{i-1}$,
which indicate the balanced state's relevance to the system's input-output behavior,
\begin{align*}
 U &= \begin{pmatrix} U_1 & U_2 \end{pmatrix}, U_1 \in \R^{N \times n}, \\
 V &= \begin{pmatrix} V_1 & V_2 \end{pmatrix}, V_1 \in \R^{N \times n},
\end{align*}
and discarding the partitions associated to small singular values $\sigma_{n+1} \ll \sigma_n$,
corresponds to the balanced truncation method \cite{morMoo81,morAnt05}.

A cross Gramian $W_X$ (\cref{def:wx}) encodes both, controllability and observability, in a single linear operator.
For a symmetric system, a balancing transformation can then be obtained from an eigenvalue decomposition (EVD) of the cross Gramian \cite{morAld91,morBauB08}:
\begin{align*}
 W_X \stackrel{\EVD}{=} U D V^\T.
\end{align*}
Alternatively, an approximate balancing transformation is obtained from an SVD of the cross Gramian \cite{morSorA02,morHimO14}:
\begin{align*}
 W_X \stackrel{\SVD}{=} U D V^\T.
\end{align*}
The truncated projections, $U_1$ and $V_1$, are obtained in the same way as for balanced truncation.
Using only the left or only the right singular vectors of $W_X$ for the (truncated) projections of the state-space,
and their transpose as reverse transformation, results in orthogonal (Galerkin) projections \cite{morHahE02a}.
This approach is called direct truncation method \cite{ForF12,morHimO14}, i.e. $V := U^\T$.

Similarly, the parameter projection can be based on associated covariance matrices.
A transformation aligning the parameters along their principal axis,
resulting from an SVD of a such parameter covariance matrix $\omega$ \cite{morSunH06,morHimO14,morHim17}:
\begin{align*}
 \omega \stackrel{\SVD}{=} \Pi \Delta \Lambda,
\end{align*}
yields truncatable projections given by the singular vectors, with partitioning of $\Pi$ and $\Lambda$ based on the singular values in $\Delta$.


\section{Empirical Gramians}\label{sec:eg} 
Classically the controllability, observability and cross Gramians are computed for linear systems by solving (linear) matrix equations.
The empirical Gramians are a data-driven extension to the classic system Gramians, and do not depend on the linear system structure.
Computing system Gramians empirically by trajectory simulations, was already motivated in \cite{morMoo81} but systematically introduced in \cite{morLalMG99}.
The central idea behind the empirical Gramians is the averaging over local Gramians for any varying quantity,
such as inputs, initial states, parameters or time-dependent components around an operating point \cite{morKeiG03}.
In the following, first the empirical Gramians for state-space input-output coherence are summarized,
then the empirical Gramians for parameter-space identifiability and combined state and parameter evaluation are described.

\subsection{State-Space Empirical Gramians} 
Gramian-based controllability and observability analysis originates in linear system theory \cite{Hes09},
which investigates linear (time-invariant) systems,
\begin{align}\label{eq:linsys}
\begin{split}
 \dot{x}(t) &= Ax(t) + Bu(t), \\
       y(t) &= Cx(t).
\end{split}
\end{align}
An obvious approach for nonlinear systems is a linearization at a steady-state \cite{morMaD88},
but this may obfuscate the original transient dynamics \cite{morSinH05,morDonSP11}.
Alternatively, the nonlinear balancing for control affine systems from \cite{morSch93},
based on controllability and observability energy functions, could be used.
Yet practically, the associated nonlinear system Gramians require solutions to a Hamilton-Jacobi partial differential equation (nonlinear controllability Gramian) and a nonlinear Lyapunov equation (nonlinear observability Gramian) or a nonlinear Sylvester equation (nonlinear cross Gramian), which is currently not feasible for large-scale systems.
A compromise between linearized and nonlinear Gramians are empirical Gramians \cite{morLalMG99,morHahE00}.

Empirical Gramians are computed by systematically averaging system Gramians obtained from numerical simulations over locations in an operating region near a steady-state.
An operating region is defined in this context by sets of perturbations for inputs / controls and the steady-state.
Originally in \cite{morLalMG99}, these perturbation sets are constructed by the Cartesian product of sets of directions (standard unit vectors),
rotations (orthogonal matrices) and scales (positive scalars) for the input and steady-state respectively.
In the empirical Gramian framework, the rotations are limited to the set of the unit matrix and negative unit matrix, as suggested in \cite{morLalMG99}.
This constraint on the rotation entails many numerical simplifications and reduces the perturbation sets to directions (standard unit vectors) and scales (non-zero scalars):
\begin{align*}
 E_u &= \{ e^m \in \R^M : m = 1 \dots M, e^m_i = \delta_{im} \}, \\
 S_u &= \{ c_k \in \R : k = 1 \dots K, c_k \neq 0 \}, \\
 E_x &= \{ \epsilon^j \in \R^N : j = 1 \dots N, \epsilon^j_i = \delta_{ij} \}, \\
 S_x &= \{ d_l \in \R : l = 1 \dots L, d_l \neq 0 \}.
\end{align*}
Yet, only single input and state components can be perturbed at a time in this manner,
which is often practically sufficient.

The original empirical Gramians use a centering of the trajectories around the temporal average and solely use impulse input type controls $u(t) = \delta(t)$ \cite{morLalMG99}.
The related empirical covariance matrices center the trajectories around a steady-state and allow arbitrary step functions $u(t) = \sum_k v_k \chi_{[t_k,t_{k+1})}(t)$, $v_k \in \R$, \mbox{$t_k \in R_{\geq 0}$}, $t_{k+1} > t_k$ \cite{morHahE02,morHahEM03}.
The empirical Gramian framework allows to compute either as well as further centering variants (\cref{sec:flags}).
In the following, empirical Gramians and empirical covariance matrices will be jointly referred to by the term ``empirical Gramian''.

\subsubsection{Empirical Controllability Gramian}\label{def:wc}
The (linear) controllability\footnote{The term \textbf{controllability} is used instead of \textbf{reachability} as in \cite{morMoo81,morFerN83,morLalMG99}.} Gramian quantifies how well the state of an underlying linear system is driven by the input and is defined as:
\begin{align*}
 W_C := \int_0^\infty \e^{At} BB^\T \e^{A^\T t} \D t = \int_0^\infty (\e^{At} B) (\e^{A t} B)^\T \D t.
\end{align*}
The empirical variant is given by the following definition based on \cite{morLalMG99,morHahE02}.
\begin{mydefine}[Empirical Controllability Gramian]
Given non-empty sets $E_u$ and $S_u$, the \textbf{empirical controllability Gramian} $\widehat{W}_C \in \R^{N \times N}$ is defined as:
\begin{align*}
 \widehat{W}_C &= \frac{1}{|S_u|} \sum_{k=1}^{|S_u|} \sum_{m=1}^{M} \frac{1}{c_k^2} \int_0^T \Psi^{km}(t) \D t, \\
 \Psi^{km}(t) &= (x^{km}(t) - \bar{x}^{km})(x^{km}(t) - \bar{x}^{km})^\T \in \R^{N \times N},
\end{align*}
with the state trajectories $x^{km}(t)$ for the input configurations \linebreak $\hat{u}^{km}(t) = c_k e^m \circ u(t) + \bar{u}$,
and the offsets $\bar{u}$, $\bar{x}^{km}$.
\end{mydefine}
For an asymptotically stable linear system, delta impulse input $u_i(t) = \delta(t)$ and an arithmetic average over time as offset $\bar{x} = \frac{1}{T} \int_0^T x(t) \D t$,
the empirical controllability Gramian is equal to the controllability Gramian $\widehat{W}_C = W_C$ \cite{morLalMG99}.

\subsubsection{Empirical Observability Gramian}\label{def:wo}
The (linear) observability Gramian quantifies how well a change in the state of an underlying linear system is visible in the outputs and is defined as:
\begin{align*}
 W_O := \int_0^\infty \e^{A^\T t} C^\T C \e^{At} \D t = \int_0^\infty (\e^{A^\T t} C^\T) (\e^{A^\T t} C^\T)^\T \D t.
\end{align*}
The empirical variant is given by the following definition based on \cite{morLalMG99,morHahE02}.
\begin{mydefine}[Empirical Observability Gramian]
Given non-empty sets $E_x$ and $S_x$, the \textbf{empirical observability Gramian} $\widehat{W}_O \in \R^{N \times N}$ is defined as:
\begin{align*}
 \widehat{W}_O &= \frac{1}{|S_x|} \sum_{l=1}^{|S_x|} \frac{1}{d_l^2} \int_0^\infty \Psi^l(t) \D t, \\
    \Psi^l_{ij}(t) &= (y^{li}(t) - \bar{y}^{li})^\T (y^{lj}(t) - \bar{y}^{lj}) \in \R,  
\end{align*}
with the output trajectories $y^{li}(t)$ for the initial state configurations \linebreak $x_0^{li} = d_l \epsilon^i + \bar{x}$,
$u(t) = \bar{u}$, and the offsets $\bar{u}$, $\bar{x}$, $\bar{y}^{li}$.
\end{mydefine}
For an asymptotically stable linear system, no input and an arithmetic average over time as offset $\bar{y} = \frac{1}{T} \int_0^T y(t) \D t$,
the empirical observability Gramian is equal to the observability Gramian $\widehat{W}_O =W_O$ \cite{morLalMG99}.

\subsubsection{Empirical Linear Cross Gramian}\label{def:wy}
The (linear) cross Gramian \cite{morFer82,morFerN83} quantifies the controllability and observability, and thus minimality, of an underlying square, \mbox{$\dim(u(t)) = \dim(y(t))$}, linear system and is defined as:
\begin{align*}
 W_X := \int_0^\infty \e^{At} BC \e^{At} \D t = \int_0^\infty (\e^{At} B) (\e^{A^\T t} C^\T)^\T \D t.
\end{align*}
Augmenting the linear system's dynamical system component with its transposed system\footnote{The transposed system is equivalent to the negative adjoint system.},
induces an associated controllability Gramian of which the upper right block corresponds to the cross Gramian \cite{morFerN85,morSha12}:
\begin{align}\label{eq:lwx}
 \begin{pmatrix} \dot{x}(t) \\ \dot{z}(t) \end{pmatrix} = \begin{pmatrix} A & 0 \\ 0 & A^\T \end{pmatrix} \begin{pmatrix} x(t) \\ z(t) \end{pmatrix} + \begin{pmatrix} B \\ C^\T \end{pmatrix} u(t) \Rightarrow \overline{W}_C = \begin{pmatrix} W_C & W_X \\ W_X^\T & W_O \end{pmatrix}.
\end{align}

The empirical variant restricted to the upper right block of this augmented controllability Gramian \eqref{eq:lwx} is given by the following definition based on \cite{morBauBHetal17}.
\begin{mydefine}[Empirical Linear Cross Gramian]
Given non-empty sets $E_u$ and $S_u$, the \textbf{empirical linear cross Gramian} \mbox{$\widehat{W}_Y \in \R^{N \times N}$} is defined as:
\begin{align*}
 \widehat{W}_Y &= \frac{1}{|S_u|} \sum_{k=1}^{|S_u|} \sum_{m=1}^{M} \frac{1}{c_k^2} \int_0^T \Psi^{km}(t) \D t, \\
 \Psi^{km}(t) &= (x^{km}(t) - \bar{x}^{km})(z^{km}(t) - \bar{z}^{km})^\T \in \R^{N \times N},
\end{align*}
with the state trajectories $x^{km}(t)$ and adjoint state trajectories $z^{km}(t)$ for the input configurations $\hat{u}^{km}(t) = c_k e^m \circ u(t) + \bar{u}$,
and the offsets $\bar{u}$, $\bar{x}^{km}$, $\bar{z}^{km}$.
\end{mydefine}
For an asymptotically stable linear system, delta impulse input $u_i(t) = \delta(t)$ and an arithmetic average over time as offset $\bar{x} = \frac{1}{T} \int_0^T x(t) \D t$, $\bar{z} = \frac{1}{T} \int_0^T z(t) \D t$, the empirical linear cross Gramian is equal to the cross Gramian due to the result of the empirical controllability Gramian.
This approach is related to balanced POD \cite{morBarDN08}.

\subsubsection{Empirical Cross Gramian}\label{def:wx}
Analytically, a cross Gramian for (control-affine) nonlinear gradient systems was developed in \cite{morIonFS11,morFujS14},
yet the computation of this nonlinear cross Gramian\footnote{Also called \textbf{cross operator} or \textbf{cross map} in this context.} is infeasible for large systems.
For (nonlinear) SISO (Single-Input-Single-Output) systems, the empirical variant of the cross Gramian is developed in \cite{StrFB06},
for (nonlinear) MIMO (Multiple-Input-Multiple-Output) systems in \cite{morHimO14}.

\begin{mydefine}[Empirical Cross Gramian]
Given non-empty sets $E_u$, $E_x$, $S_u$ and $S_x$, the \textbf{empirical cross Gramian} $\widehat{W}_X \in \R^{N \times N}$ is defined as:
\begin{align*}
  \widehat{W}_X &= \frac{1}{|S_u||S_x| M} \sum_{k=1}^{|S_u|} \sum_{l=1}^{|S_x|} \sum_{m=1}^M \frac{1}{c_k d_l} \int_0^\infty \Psi^{klm}(t) \D t, \\
\Psi^{klm}_{ij} &= (x_i^{km}(t) - \bar{x}_i^{km})(y_m^{lj}(t) - \bar{y}_m^{lj}) \in \R,
\end{align*}
with the state trajectories $x^{km}(t)$ for the input configurations \linebreak $\hat{u}^{km}(t) = c_k e^m \circ u(t) + \bar{u}$,
the output trajectories $y^{lj}(t)$ for the initial state configurations $x_0^{lj} = d_l \epsilon^j + \bar{x}$,
and the offsets $\bar{u}$, $\bar{x}^{km}$, $\bar{x}$, $\bar{y}^{lj}$.
\end{mydefine}
For an asymptotically stable linear system, delta impulse input $u(t) = \delta(t)$ and an arithmetic averages over time as offsets $\bar{x} = \frac{1}{T} \int_0^T x(t) \D t$, $\bar{y} = \frac{1}{T} \int_0^T y(t) \D t$,
the empirical cross Gramian is equal to the cross Gramian $\widehat{W}_X = W_X$ \cite{StrFB06,morHimO14}.

\subsubsection{Empirical Non-Symmetric Cross Gramians}\label{def:wz}
The (empirical) cross Gramian is only computable for square systems,
and verifiably useful for symmetric or gradient systems \cite{morFer82,morSorA02,morHimO14}.
In \cite{morHimO16} an extension to the classic cross Gramian is proposed.
Based on results from decentralized control \cite{MoaK06}, a non-symmetric cross Gramian is computable for non-square systems and thus non-symmetric systems.
Given a partitioning of $B = [b_1, \dots, b_M]$, $b_i \in \R^{M \times 1}$ and $C = [c_1, \dots, c_Q]^\T$, $c_j \in \R^{1 \times Q}$ from the linear system \eqref{eq:linsys}, the (linear) non-symmetric cross Gramian is defined as:
\begin{align*}
 W_Z := \sum_{i=1}^M \sum_{j=1}^Q \int_0^\infty \e^{At} b_i c_j \e^{At} \D t = \int_0^\infty \e^{At} \Big(\sum_{i=1}^M b_i\Big) \Big(\sum_{j=1}^Q c_j\Big) \e^{At} \D t.
\end{align*}
For this cross Gramian to the associated ``average'' SISO system, an empirical variant is then given by:
\begin{mydefine}[Empirical Non-Symmetric Cross Gramian]
Given non-empty sets $E_u$, $E_x$, $S_u$ and $S_x$, the \textbf{empirical non-symmetric cross Gramian} \mbox{$\widehat{W}_Z \in \R^{N \times N}$} is defined as:
\begin{align*}
  \widehat{W}_Z &= \frac{1}{|S_u||S_x| M} \sum_{k=1}^{|S_u|} \sum_{l=1}^{|S_x|} \sum_{m=1}^M \sum_{q=1}^Q \frac{1}{c_k d_l} \int_0^\infty \Psi^{klmq}(t) \D t, \\
\Psi^{klmq}_{ij} &= (x_i^{km}(t) - \bar{x}_i^{km}) (y_q^{lj}(t) - \bar{y}_q^{lj}) \in \R,
\end{align*}
with the state trajectories $x^{km}(t)$ for the input configurations \linebreak $\hat{u}^{km}(t) = c_k e^m \circ u(t) + \bar{u}$,
the output trajectories $y^{lj}(t)$ for the initial state configurations $x_0^{lj} = d_l \epsilon^j + \bar{x}$,
and the offsets $\bar{u}$, $\bar{x}^{km}$, $\bar{x}$, $\bar{y}^{lj}$.
\end{mydefine}

\begin{mycorollary}
For an asymptotically stable linear system, delta impulse input $u(t) = \delta(t)$ and arithmetic averages over time as offsets $\bar{x} = \frac{1}{T} \int_0^T x(t) \D t$, $\bar{y} = \frac{1}{T} \int_0^T y(t) \D t$,
the empirical non-symmetric cross Gramian is equal to the cross Gramian \linebreak \mbox{$\widehat{W}_Z = W_Z$} of the average SISO system.
\end{mycorollary}

\begin{proof}
This is a direct consequence of \cite[Lemma.~3]{morHimO14}.
\end{proof}

\subsection{Parameter-Space Empirical Gramians}
To transfer the idea of Gramian-based state-space reduction to parameter-space reduction,
the concepts of \textit{controllability} and \textit{observability} are extended to the parameter-space.
This leads to controllability-based parameter identification and observability-based parameter identification.

\subsubsection{Empirical Sensitivity Gramian}\label{def:ws}
Controllability-based parameter identification can be realized using an approach from \cite{morSunH06},
which treats the parameters as (additional) constant inputs.
The controllability Gramian for a linear system with linear parametrization (constant source or load) can be decomposed additively based on linear superposition:
\begin{align*}
 \dot{x}(t) &= Ax(t) + Bu(t) + F\theta = Ax(t) + \begin{pmatrix} B & F \end{pmatrix} \begin{pmatrix} u(t) \\ \theta \end{pmatrix} \\ &\Rightarrow W_C = W_C(A,B) + \sum_{i=1}^P W_{C,i}(A,F_{*i}) .
\end{align*}
Similar to \cite{StrFB06}, the trace of the parameter-controllability Gramians $W_{C,i}$ embodies a measure of (average) sensitivity,
and holds approximately for systems with nonlinear parametrization \cite{morHimO13}.
\begin{mydefine}[Empirical Sensitivity Gramian]
The \textbf{empirical sensitivity Gramian} is given by a diagonal matrix with entries corresponding to the traces of the parameter-controllability Gramians,
\begin{align*}
 W_{S,ii} := \tr(W_{C,i}).
\end{align*}
\end{mydefine}
The sum over all controllability Gramians can also be used for robust model reduction \cite{morSunH06a}.
Similarly, treating the parameters as inputs, the cross Gramian's trace can be utilized as a sensitivity measure \cite{LysHA16}.

\subsubsection{Empirical Identifiability Gramian}\label{def:wi}
For an observability-based parameter identification, the parameters are interpreted as additional states of the system \cite{SinH05,GefFSetal08}.
This approach leads to the \textbf{augmented system}, in which the system's state $x$ is appended by the parameter $\theta$ and,
since the parameters are (assumed) constant over time, the components of the vector-field associated to the \textit{parameter-states} are zero:
\begin{align}\label{eq:augsys}
\begin{split}
 \begin{pmatrix} \dot{x}(t) \\ \dot{\theta}(t) \end{pmatrix} &= \begin{pmatrix} f(t,x(t),u(t),\theta(t)) \\ 0 \end{pmatrix}, \\
                       y(t) &= g(t,x(t),u(t),\theta(t)), \\
 \begin{pmatrix} x(0) \\ \theta(0) \end{pmatrix} &= \begin{pmatrix} x_0 \\ \theta \end{pmatrix},
\end{split}
\end{align}
leaving the parameter-state's initial value for testing parameter perturbations.
The observability Gramian to this augmented system, the \textit{augmented observability Gramian},
has the block structure:
\begin{align*}
 \widehat{W}_O = \begin{pmatrix} W_O & W_M \\ W_M^\T & W_P \end{pmatrix},
\end{align*}
with the state-space observability Gramian $W_O$, the parameter-space observability Gramian $W_I$ and the mixed state and parameter block $W_M = W_M^\T$.
To isolate the parameter identifiability information, the state-space block is eliminated.
\begin{mydefine}[Empirical Identifiability Gramian]
The \textbf{empirical identifiability Gramian} is given by the Schur-complement of the empirical augmented observability Gramian for the lower right block:
\begin{align*}
 W_I = W_P - W_M^\T W_O^{-1} W_M.
\end{align*}
\end{mydefine}
Often, it is sufficient to approximate the empirical identifiability Gramian by the lower right block of the augmented observability Gramian $W_P$:
\begin{align*}
 W_I \approx W_P.
\end{align*}
Apart from the relation of identifiability Gramian to the Fischer information matrix \cite{SinH05}, 
also the connection of the (parameter) observability Gramian to the (parameter) Hessian matrix \cite{morLieFWetal13} is noted here.

\subsubsection{Empirical Cross-Identifiability Gramian}\label{def:wj}
If a system is square, the augmented system \eqref{eq:augsys} remains square and for linear systems also symmetry is preserved.
Hence, a cross Gramian of the augmented system is computable \cite{morHimO14}.
\begin{mydefine}[Empirical Joint Gramian]
The \textbf{empirical joint Gramian} is given by the empirical cross Gramian of the augmented system.
\end{mydefine}
The joint Gramian is an augmented cross Gramian and has a similar block structure as the augmented observability Gramian,
\begin{align*}
 W_J = \begin{pmatrix} W_X & W_m \\ 0 & 0 \end{pmatrix},
\end{align*}
but due to the uncontrollable parameter-states, the lower (parameter-related) blocks are identically zero.
Nonetheless, the observability-based parameter identifiability information can be extracted from the mixed block $W_m$.
\begin{mydefine}[Empirical Cross-Identifiability Gramian]
The \textbf{empirical cross-identifiability Gramian} is given by the Schur-complement of the symmetric part of the empirical joint Gramian for the lower right block:
\begin{align*}
 W_{\ddot{I}} = 0 - \frac{1}{2} W_m^\T (W_X + W_X^\T)^{-1} W_m.
\end{align*}
\end{mydefine}
Thus, the empirical joint Gramian enables the combined state and parameter analysis by the empirical cross Gramian $W_X$ and empirical cross-identifiability Gramian $W_{\ddot{I}}$ from a single $N \times N+P$ matrix.
Note, that the empirical joint Gramian may also be computed based on the non-symmetric cross Gramian.
Additionally, a parameter Gramian, such as $W_I$ or the $W_{\ddot{I}}$, could be balanced with the loadability Gramian from \cite{morTolA17}.

\subsection{Notes on Empirical Gramians}
It should be noted that empirical Gramians only yield workable results if the operating region of the system is restricted to a single steady-state.
If the trajectories used for the assembly of empirical Gramians are periodic or do not attain this steady-state, their performance is similar to time-limited balancing methods, see for example \cite{morJazST15} and references therein.
Overall, the quality of empirical-gramian-based methods depends largely on the quality of the measured or simulated (output) trajectory data.
Yet, due to the data-driven nature of the empirical Gramians, even unstable systems or systems with inhomogeneous initial conditions are admissible.


\section{Implementation Details}\label{sec:impl}
This section states concisely the theoretical, practical and technical design decisions in the implementation of the empirical Gramian framework - \emgr \cite{morHim18a}, as well as describing the unified and configurable approach to empirical Gramian computation.

\subsection{Design Principles} 
\emgr is realized using the high-level, interpreted Matlab programming language,
which is chosen due to its widespread use, long-term compatibility and mathematical expressiveness.
This enables first, a wide circulation due to compatibility with \textsc{Mathworks Matlab}\textsuperscript{\textregistered} \cite{matlabweb},
and second, the usage of the open-source variant\footnote{The \textsc{FreeMat} Matlab interpreter (Version: 4.2) \cite{freematweb} is not compatible.} \textsc{Gnu Octave} \cite{octaveweb}. 
Generally, the implementation of \emgr follows the procedural programming paradigm, includes various functional programming techniques and avoids object-oriented programming.

Since empirical Gramians are computable by mere basic linear algebra operations,
Matlab code can be evaluated efficiently by vectorization,
which transfers computationally intensive tasks as bulk operations to the BLAS (Basic Linear Algebra Subroutines) back-end.

Overall, \emgr is a reusable open-source toolbox,
and encompasses less than $500$ LoC (Lines of Code) in a single file and a cyclomatic complexity of $<100$ of the main function.
Apart from a Matlab interpreter, \emgr has \textbf{no} further dependencies, such as on other toolboxes.
The source code is engineered with regard to the best practice guides \cite{Joh11} (coding style) and \cite{Alt15} (performance).

Furthermore, two variants of \emgr are maintained:
First, \texttt{emgr$\_$oct}\footnote{See \url{http://gramian.de/emgr_oct.m}}, uses \textsc{Octave}-specific language extension: \textit{default arguments} and \textit{assignment operations},
second, \texttt{emgr$\_$lgc}\footnote{See \url{http://gramian.de/emgr_lgc.m}}, enables compatibility to \textsc{Matlab} versions before 2016b not supporting \textit{implicit expansion}, also known as \textit{automatic broadcasting}.

\subsection{Parallelization}
Apart from vectorization allowing the implicit use of SIMD (single-instruction-multiple-data) functionality for vectorized block operations,
also multi-core parallelization is used to maximize use of available compute resources.

\subsubsection{Shared Memory Parallelization} 
For shared memory systems with uniform memory access (UMA), two types of parallelization are utilized.
First, an \textbf{implicit parallelization} may be applied by the interpreter for an additional acceleration of block operations.
Second, \textbf{explicit parallelization} is available for the computation of different state and output trajectories,
using parallel for-loops \texttt{parfor}\footnote{See \url{http://mathworks.com/help/matlab/ref/parfor.html} },
but deactivated by default to guarantee replicable results, as the use of \texttt{parfor} does not guarantee a unique order of execution.

\subsubsection{Heterogeneous Parallelization}\label{sec:gpgpu} 
The actual empirical Gramians result from $N^2$ inner products.
In case of the default Euclidean inner product, this amounts to a dense matrix-matrix-multiplication\footnote{Implemented as GEMM (Generalized Matrix Multiplication $R = AB+\alpha C$) by BLAS.},
which can be efficiently computed by GPGPUs (General Purpose Graphics Processing Units).
In case an integrated GPU with zero-copy shared memory architecture, such as UMM (uniform memory model) or hUMA (heterogeneous unified memory access) \cite{RogMM13}, is used, 
the assembly of the Gramian matrices can be performed with practically no overhead, since the trajectories,
which are usually computed and stored in CPU memory space,
do not need to be copied between CPU and GPU memory spaces.
The GPU can directly operate on the shared memory.

\subsubsection{Distributed Memory Parallelization}\label{def:dwx}
A disadvantage of empirical Gramian matrices is the quadratically growing memory requirements with respect to the state-space (and parameter-space) dimension, since for an $N$ dimensional system, a (dense) empirical Gramian of dimension $N \times N$ is computed.
To combat this shortcoming, a specific property of the empirical cross Gramian can be exploited:
The columns of the empirical cross Gramian, and thus the empirical joint Gramian, may be computed separately,
\begin{align*}
  \widehat{W}_X &= [\widehat{w}_X^1, \dots, \widehat{w}_X^N], \\
  \widehat{w}_X^j &= \frac{1}{|S_u||S_x| M} \sum_{k=1}^{|S_u|} \sum_{l=1}^{|S_x|} \sum_{m=1}^M \frac{1}{c_k d_l} \int_0^\infty \psi^{klmj}(t) \D t \in \R^{N \times 1}, \\
\psi^{klmj}_i &= (x_i^{km}(t) - \bar{x}_i^{km})(y_m^{lj}(t) - \bar{y}_m^{lj}) \in \R,
\end{align*}
hence this \textbf{distributed empirical cross Gramian} \cite[Sec.~4.2]{morHimLR16} is computable in parallel and communication-free on a distributed memory computer system,
or sequentially in a memory-economical manner as a \textbf{low-rank empirical cross Gramian} \cite{morHimLRetal17} on a unified memory computer system.
This column-wise computability translates also to the empirical joint Gramian and the non-symmetric variants of the empirical cross and joint Gramian. 

Based on this partitioning, a low-rank representation can be obtained in a memory-bound or compute-bound setting
together with the HAPOD (hierarchical approximate proper orthogonal decomposition) \cite{morHimLR16}.
This POD variant allows to directly compute a Galerkin projection from an arbitrary column-wise partitioning of the empirical cross Gramian.


\section{Interface}\label{sec:intf} 
\emgr provides a uniform function call for the computation of all empirical Gramian types.
The subsequent signature documentation is based on \cite{morHim17} and \url{http://gramian.de} \footnote{The current instance of \url{http://gramian.de} is preserved at \url{http://archive.is/bOBpW}.}.
Minimally, the \emgr function requires five mandatory arguments (single letter):
\begin{center}
 \texttt{emgr(f,g,s,t,w)}
\end{center}
additionally eight optional arguments (double letter) allow a usage by:
\begin{center}
 \texttt{emgr(f,g,s,t,w,pr,nf,ut,us,xs,um,xm,dp)}
\end{center}
furthermore, a single argument variant may also be used,
\begin{center}
 \texttt{emgr(\textquotesingle version\textquotesingle)}
\end{center}
which returns the current version number.

\subsection{Mandatory Arguments} 
For the minimal usage, the following five arguments are required:

\begin{itemize}

 \item[\texttt{f}] handle to a function with the signature \texttt{xdot = f(x,u,p,t)} representing the system's vector-field and expecting the arguments: current state \texttt{x}, current input \texttt{u}, (current) parameter \texttt{p} and current time \texttt{t}. \vskip1ex

 \item[\texttt{g}] handle to a function with the signature \texttt{y = g(x,u,p,t)} representing the system's output functional and expecting the arguments: current state \texttt{x}, current input \texttt{u}, (current) parameter \texttt{p} and current time \texttt{t}.

If \texttt{g~=~1}, the identity output functional $g(t,x(t),u(t),\theta) = x(t)$ is assumed. \vskip1ex

 \item[\texttt{s}] three component vector \texttt{s = [M,N,Q]} setting the dimensions of the input $M~:=~\dim(u(t))$, state $N~:=~\dim(x(t))$ and output $Q~:=~\dim(y(t))$. \vskip1ex

 \item[\texttt{t}] two component vector \texttt{t = [h,T]} specifying the time-step width $h$ and time horizon $T$. \vskip1ex

 \item[\texttt{w}] character selecting the empirical Gramian type; for details see \cref{sec:features}.

\end{itemize}

\subsection{Features}\label{sec:features} 
The admissible characters to select the empirical Gramian type are as follows:

\begin{itemize}

 \item[\texttt{\textquotesingle c\textquotesingle}] Empirical controllability Gramian (see \cref{def:wc}), \\ \emgr returns a matrix:
  \begin{itemize}
   \item[] $N \times N$ empirical controllability Gramian matrix $W_C$. 
  \end{itemize}

 \item[\texttt{\textquotesingle o\textquotesingle}] Empirical observability Gramian (see \cref{def:wo}), \\ \emgr returns a matrix:
  \begin{itemize}
   \item[] $N \times N$ empirical observability Gramian matrix $W_O$.
  \end{itemize}

\pagebreak

 \item[\texttt{\textquotesingle x\textquotesingle}] Empirical cross Gramian (see \cref{def:wx}), \\ \emgr returns a matrix:
  \begin{itemize}
   \item[] $N \times N$ empirical cross Gramian matrix $W_X$.
  \end{itemize}

 \item[\texttt{\textquotesingle y\textquotesingle}] Empirical linear cross Gramian (see \cref{def:wy}), \\ \emgr returns a matrix:
  \begin{itemize}
   \item[] $N \times N$ empirical linear cross Gramian matrix $W_Y$.
  \end{itemize}

 \item[\texttt{\textquotesingle s\textquotesingle}] Empirical sensitivity Gramian (see \cref{def:ws}), \\ \emgr returns a cell array\footnote{A cell array is a generic container (array) in the Matlab language.} holding:
  \begin{itemize}
   \item[] $N \times N$ empirical controllability Gramian matrix $W_C$,
   \item[] $P \times 1$ empirical sensitivity Gramian diagonal $W_S$.
  \end{itemize}

 \item[\texttt{\textquotesingle i\textquotesingle}] Empirical identifiability Gramian (see \cref{def:wi}), \\ \emgr returns a cell array holding:
  \begin{itemize}
   \item[] $N \times N$ empirical observability Gramian matrix $W_O$,
   \item[] $P \times P$ empirical identifiability Gramian matrix $W_I$.
  \end{itemize}

 \item[\texttt{\textquotesingle j\textquotesingle}] Empirical joint Gramian (see \cref{def:wj}), \\ \emgr returns a cell array holding:
  \begin{itemize}
   \item[] $N \times N$ empirical cross Gramian matrix $W_X$,
   \item[] $P \times P$ empirical cross-identifiability Gramian matrix $W_{\ddot{I}}$.
  \end{itemize}

\end{itemize}

\subsubsection{Non-Symmetric Cross Gramian} 
The non-symmetric cross Gramian \cite{morHimO16} (see \cref{def:wz}) is a special variant of the cross Gramian for non-square and non-symmetric MIMO systems, which reduces to the regular cross Gramian for SISO systems.
Since the computation is similar to the empirical cross Gramian and a non-symmetric variant of the empirical joint Gramian shall be computable too,
instead of a Gramian type selected through the argument \texttt{w}, it is selectable via an option flag:
Non-symmetric variants may be computed for the empirical cross Gramian (\texttt{w~=~\textquotesingle x\textquotesingle}),
empirical linear cross Gramian (\texttt{w~=~\textquotesingle y\textquotesingle}) or
empirical joint Gramian (\texttt{w~=~\textquotesingle j\textquotesingle}) by activating the flag \texttt{nf(7)~=~1}.

\subsubsection{Parametric Systems} 
Parametric model order reduction is accomplished by averaging an empirical Gramian over a discretized parameter-space \cite{morHimO15a}.
To this end the parameter sampling points, arranged as columns of a matrix, can be supplied via the optional argument \texttt{pr}.

\subsubsection{Time-Varying Systems}
Since empirical Gramians are purely based on trajectory data,
they are also computable for time varying systems as described in \cite{morConI04}.
The empirical Gramian framework can compute averaged Gramians for time varying systems \cite{morNilR09a}, which are time independent matrices.

\subsection{Optional Arguments} 
The eight optional arguments allow a detailed definition of the operating region and configuration of the computation.


\begin{itemize}

 \item[\texttt{pr}] system parameters (Default value: \texttt{0})
  \begin{itemize}
   \item[\textit{vector}] a column vector holding the parameter components,
   \item[\textit{matrix}] a set of parameters, each column holding one parameter.
  \end{itemize}

 \item[\texttt{nf}] twelve component vector encoding the option flags,
                    for details see \cref{sec:flags}.

 \item[\texttt{ut}] input function (Default value: \texttt{1})
  \begin{itemize}
   \item[\textit{handle}] function handle expecting a signature \texttt{u$\_$t = u(t)},
   \item[$0$~~~~]         pseudo-random binary input,
   \item[$1$~~~~]         delta impulse input,
   \item[$\infty$~~~]     decreasing frequency exponential chirp.
  \end{itemize}

 \item[\texttt{us}] steady-state input (Default value: \texttt{0})
  \begin{itemize}
   \item[\textit{scalar}] set all $M$ steady-state input components to argument,
   \item[\textit{vector}] set steady-state input to argument of expected dimension $M \times 1$.
  \end{itemize}

 \item[\texttt{xs}] steady-state (Default value: \texttt{0})
  \begin{itemize}
   \item[\textit{scalar}] set all $N$ steady-state components to argument,
   \item[\textit{vector}] set steady-state to argument of expected dimension $N \times 1$.
  \end{itemize}

 \item[\texttt{um}] input scales (Default value: \texttt{1})
  \begin{itemize}
   \item[\textit{scalar}] set all $M$ maximum input scales to argument,
   \item[\textit{vector}] set maximum input scales to argument of expected dimension $M \times 1$,
   \item[\textit{matrix}] set scales to argument with $M$ rows; used as is.
  \end{itemize}

 \item[\texttt{xm}] initial state scales (Default value: \texttt{1})
  \begin{itemize}
   \item[\textit{scalar}] set all $N$ maximum initial state scales to argument,
   \item[\textit{vector}] \mbox{set maximum steady-state scales to argument of expected dimension $N \times 1$},
   \item[\textit{matrix}] set scales to argument with $N$ rows; used as is.
  \end{itemize}

 \item[\texttt{dp}] inner product interface via a handle to a function with the signature \mbox{\texttt{z = dp(x,y)}} defining the dot product for the Gramian matrix computation (Default value: \texttt{[]}).

\end{itemize}

\subsubsection{Inner Product Interface}
The empirical Gramian matrices are computed by inner products of trajectory data.
A custom inner products for the assembly of the empirical Gramians matrix,
can be set by the argument \texttt{dp}, which expects a handle to a function with the signature: 
\begin{center}
 \texttt{z = dp(x,y)}
\end{center}


and the arguments:
\begin{itemize}

 \item[\texttt{x}] matrix of dimension $N \times \frac{T}{h}$,

 \item[\texttt{y}] matrix of dimension $\frac{T}{h} \times n$ for $n \leq N$.

\end{itemize}
The return value $z$ is typically an $N \times n$ matrix,
but scalar or vector-valued $z$ are admissible, too.
By default, the Euclidean inner product, the standard matrix multiplication, is used:
\begin{center}
 \texttt{dp = @(x,y) mtimes(x,y)}.
\end{center}
Other choices are for example:
covariance-weighted products for Gaussian-noise-driven systems yielding system covariances \cite{NikG69},
reproducing kernel Hilbert spaces (RKHS) \cite{morBouH17}, such as the polynomial, Gaussian or Sigmoid kernels \cite{FasM15},
or energy-stable inner products \cite{morKalBAetal14}.
Also, weighted Gramians \cite{morChoSB16} and time-weighted system Gramians \cite{morSchD95} can be computed using this interface, i.e.:
\begin{center}
 \texttt{dp = @(x,y) mtimes([0:h:T].\^{}k.*x,y)}
\end{center}
for a monomial of order $k$ time-domain weighted inner product.
Furthermore, the inner product interface may be used to directly compute the trace of an empirical Gramian by using a pseudo-kernel: 
\begin{center}
 \texttt{dp = @(x,y) sum(sum(x.*y\textquotesingle))}
\end{center}
which exploits a property for computing the trace of a matrix product \linebreak \mbox{$\tr(AB) = \sum_i \sum_j A_{ij} B_{ji}$}.
This interface may also be used to compute only the empirical Gramian's diagonal:
\begin{center}
 \texttt{dp = @(x,y) sum(x.*y\textquotesingle,2)}
\end{center}
for input-output importance \cite{morSnoVT17} or input-output coherence \cite[Ch.~13]{morFer82}.
Lastly, it is noted that offloading matrix multiplications to an accelerator such as a GPGPU, motivated in \cref{sec:gpgpu}, can also be achieved using this interface.

\subsection{Option Flags}\label{sec:flags}
The vector \texttt{nf} contains twelve components,
each representing an option with the default value zero and the following functionality:

\pagebreak

\begin{itemize}

 \item[\texttt{nf(1)}] Time series centering:
  \begin{itemize}
   \item[\texttt{=~0}] No centering,
   \item[\texttt{=~1}] Steady-state (for empirical covariance matrices),
   \item[\texttt{=~2}] Final state,
   \item[\texttt{=~3}] Arithmetic average over time (for empirical Gramians),
   \item[\texttt{=~4}] Root-mean-square over time,
   \item[\texttt{=~5}] Mid-range over time.
  \end{itemize}

 \item[\texttt{nf(2)}] Input scale sequence:
  \begin{itemize}
   \item[\texttt{=~0}] Single scale: \texttt{um $\leftarrow$ um},
   \item[\texttt{=~1}] Linear scale subdivision: \texttt{um $\leftarrow$ um * [0.25, 0.5, 0.75, 1.0]},
   \item[\texttt{=~2}] Geometric scale subdivision: \texttt{um $\leftarrow$ um * [0.125, 0.25, 0.5, 1.0]},
   \item[\texttt{=~3}] Logarithmic scale subdivision: \texttt{um $\leftarrow$ um * [0.001, 0.01, 0.1, 1.0]},
   \item[\texttt{=~4}] Sparse scale subdivision: \texttt{um $\leftarrow$ um * [0.01, 0.5, 0.99, 1.0]}.
  \end{itemize}


 \item[\texttt{nf(3)}] Initial state scale sequence:
  \begin{itemize}
   \item[\texttt{=~0}] Single scale: \texttt{xm $\leftarrow$ xm},
   \item[\texttt{=~1}] Linear scale subdivision: \texttt{xm $\leftarrow$ xm * [0.25, 0.5, 0.75, 1.0]},
   \item[\texttt{=~2}] Geometric scale subdivision: \texttt{xm $\leftarrow$ xm * [0.125, 0.25, 0.5, 1.0]},
   \item[\texttt{=~3}] Logarithmic scale subdivision: \texttt{xm $\leftarrow$ xm * [0.001, 0.01, 0.1, 1.0]},
   \item[\texttt{=~4}] Sparse scale subdivision: \texttt{xm $\leftarrow$ xm * [0.01, 0.5, 0.99, 1.0]}.
  \end{itemize}

 \item[\texttt{nf(4)}] Input directions:
  \begin{itemize}
   \item[\texttt{=~0}] Positive and negative: \texttt{um $\leftarrow$ [-um, um]},
   \item[\texttt{=~1}] Only positive: \texttt{um $\leftarrow$ um}.
  \end{itemize}

 \item[\texttt{nf(5)}] Initial state directions:
  \begin{itemize}
   \item[\texttt{=~0}] Positive and negative: \texttt{xm $\leftarrow$ [-xm, xm]},
   \item[\texttt{=~1}] Only positive: \texttt{xm $\leftarrow$ xm}.
  \end{itemize}

 \item[\texttt{nf(6)}] Normalizing:
  \begin{itemize}
   \item[\texttt{=~0}] No normalization,
   \item[\texttt{=~1}] Scale with Gramian diagonal (see \cite{EbeA12}),
   \item[\texttt{=~2}] Scale with steady-state (see \cite{morSunH06b}).
  \end{itemize}

 \item[\texttt{nf(7)}] Non-Symmetric Cross Gramian, only $W_X$, $W_Y$, $W_J$: 
  \begin{itemize}
   \item[\texttt{=~0}] Regular cross Gramian,
   \item[\texttt{=~1}] Non-symmetric cross Gramian.
  \end{itemize}

 \item[\texttt{nf(8)}] Extra input for state and parameter perturbation trajectories, only \mbox{$W_O$, $W_X$, $W_S$, $W_I$, $W_J$}:
  \begin{itemize}
   \item[\texttt{=~0}] No extra input,
   \item[\texttt{=~1}] Apply extra input (see \cite{PowM15}).
  \end{itemize}

\pagebreak

 \item[\texttt{nf(9)}] Center parameter scales, only $W_S$, $W_I$, $W_J$:
  \begin{itemize}
   \item[\texttt{=~0}] No centering,
   \item[\texttt{=~1}] Center around arithmetic mean,
   \item[\texttt{=~2}] Center around logarithmic mean.
  \end{itemize}


 \item[\texttt{nf(10)}] Parameter Gramian variant, only $W_S$, $W_I$, $W_J$:
  \begin{itemize}
   \item[\texttt{=~0}] Average input-to-state ($W_S$), detailed Schur-complement ($W_I$, $W_J$),
   \item[\texttt{=~1}] Average input-to-output ($W_S$), approximate Schur-complement \mbox{($W_I$, $W_J$).}
  \end{itemize}

 \item[\texttt{nf(11)}] Empirical cross Gramian partition width, only $W_X$, $W_J$: 
  \begin{itemize}
   \item[\texttt{=~0}] Full cross Gramian computation, no partitioning.
   \item[\texttt{<~N}] Maximum partition size in terms of cross Gramian columns.
  \end{itemize}

 \item[\texttt{nf(12)}] Partitioned empirical cross Gramian running index, only $W_X$, $W_J$:
  \begin{itemize}
   \item[\texttt{=~0}] No partitioning.
   \item[\texttt{>~0}] Index of the set of cross Gramian columns to be computed.
  \end{itemize}

\end{itemize}

\subsubsection{Schur-Complement} 
The observability-based parameter Gramians, the empirical identifiability Gramian $W_I$ and the empirical cross-identifiability Gramian $W_{\ddot{I}}$,
utilize an inversion as part of a (approximated) Schur-complement.
Instead of using a Schur-complement solver or the pseudo-inverse, an approximate inverse with computational complexity $\mathcal{O}(N^2)$ based on \cite{WuYVetal13} is utilized,
\begin{align*}
 A^{-1} \approx D^{-1} - D^{-1} E D^{-1},
\end{align*}
with the diagonal matrix $D$, $D_{ii} = A_{ii}$ and the matrix of off-diagonal elements $E = A - D$.
This approximate inverse is used by default for $W_I$ and $W_{\ddot{I}}$.

\subsubsection{Partitioned Computation}
The partitioned empirical cross Gramian (\cref{def:dwx}) can be configured by the option flags \texttt{nf(11)} and \texttt{nf(12)},
with \texttt{nf(11)} defining the maximum number of columns per partition, and \texttt{nf(12)} setting the running index of the current partition.  
Together with a partitioned singular value decomposition, such as the HAPOD \cite{morHimLR16},
an empirical-cross-Gramian-based Galerkin projection is computable with minimal communication parallely on a distributed memory system, or sequentially on a shared memory system \cite{morHimLRetal17}.

\subsection{Solver Configuration}\label{sec:rk}
To provide a problem-specific integrator to generate the state and output trajectories,
a global variable named \texttt{ODE} is available, and expects a handle to a function with the signature:
\begin{center}
 \texttt{y = ODE(f,g,t,x0,u,p)}
\end{center}
comprising the arguments:
\begin{itemize}

 \item[\texttt{f}] handle to a function with the signature \texttt{xdot = f(x,u,p,t)} representing the system's vector-field and expecting the arguments: current state \texttt{x}, current input \texttt{u}, (current) parameter \texttt{p} and current time \texttt{t}. \vskip1ex

 \item[\texttt{g}] handle to a function with the signature \texttt{y = g(x,u,p,t)} representing the system's output functional and expecting the arguments: the current state \texttt{x}, current input \texttt{u}, (current) parameter \texttt{p} and current time \texttt{t}. \vskip1ex

 \item[\texttt{t}] two component vector \texttt{t = [h,T]} specifying the time-step width $h$ and time horizon $T$. \vskip1ex

 \item[\texttt{x0}] column vector of dimension $N$ setting the initial condition. \vskip1ex

 \item[\texttt{u}] handle to a function with the signature \texttt{u$\_$t = u(t)}. \vskip1ex

 \item[\texttt{p}] column vector of dimension $P$ holding the (current) parameter.

\end{itemize}
The solver is expected to return a discrete trajectory matrix of dimension \mbox{$\dim(g(x(t),u(t),\theta,t)) \times \frac{T}{h}$}.

As a default solver for (nonlinear) initial value problems,
the optimal second-order strong stability preserving (SSP) explicit Runge-Kutta method \cite{Ket08} is included in \emgr.
This single-step integrator is implemented in a low-storage variant, 
and the stability of this method can be increased by additional stages, which is configurable by a global variable named \texttt{STAGES}.
The default number of stages is \mbox{\texttt{STAGES} $= 3$}, yielding the SSP32 method.

\subsection{Sample Usage}
To illustrate the usage of \emgr,
the Matlab code for the computation of the empirical cross Gramian of a small linear system is presented in \cref{fig:code}.
For demonstration purposes, this system has one uncontrollable and unobservable, one controllable and unobservable, one uncontrollable and observable, and one controllable and observable state:
\begin{align*}
 A := -\frac{1}{2} \1, \quad B := \begin{pmatrix} 0 & 1 & 0 & 1 \end{pmatrix}^\T, \quad C := \begin{pmatrix} 0 & 0 & 1 & 1 \end{pmatrix}.
\end{align*}
The cross Gramian computes as:
\begin{align*}
 A W_X + W_X A = BC \Rightarrow W_X = BC,
\end{align*}
which is approximately computed empirically following \cref{def:wx} in \cref{fig:code}.

\begin{figure}[h!]
\begin{minipage}{\textwidth}\texttt{%
M = 1; ~~~~~~~~~~~~~~~~~~~~~~~~~~~\% Number of inputs \\
N = 4; ~~~~~~~~~~~~~~~~~~~~~~~~~~~\% Number of states \\
Q = 1; ~~~~~~~~~~~~~~~~~~~~~~~~~~~\% Number of outputs \\
A = -0.5*eye(N); ~~~~~~~~~~~~~~~~~\% System matrix \\
B = [0;1;0;1]; ~~~~~~~~~~~~~~~~~~~\% Input matrix \\
C = [0,0,1,1]; ~~~~~~~~~~~~~~~~~~~\% Output matrix \\
f = @(x,u,p,t) A*x + B*u; ~~~~~~~~\% Vector field \\
g = @(x,u,p,t) C*x; ~~~~~~~~~~~~~~\% Output functional \\
h = 0.1; ~~~~~~~~~~~~~~~~~~~~~~~~~\% Time step size \\
T = 10.0; ~~~~~~~~~~~~~~~~~~~~~~~~\% Time horizon \\
Wx = emgr(f,g,[M,N,Q],[h,T],\textquotesingle x\textquotesingle); \% $\approx$ B*C
}
\end{minipage}
\caption{Sample code for the computation of the empirical cross Gramian of a non-minimal fourth order system.}
\label{fig:code}
\end{figure}


\section{Numerical Examples}\label{sec:numex}
In this section empirical-Gramian-based model reduction techniques are demonstrated for three test systems using \cite{octaveweb};
first, for a linear state-space symmetric MIMO system, 
second, for a hyperbolic SISO system,
and third for a nonlinear SIMO system.
All numerical tests are performed using \textsc{Octave~4.4} \cite{octave}.

\subsection{Linear Verification}\label{sc:numex1}
The first example is a linear test system of the form \eqref{eq:linsys} and generated using the inverse Lyapunov procedure \cite{morWiki},
in a variant that enforces state-space symmetric systems \cite{morHimO17}.
For state-space symmetric systems, $A = A^\T$, $B = C^\T$, all system Gramians are symmetric and equal \cite{morLiuST98}.
The system is configured to have $N = \dim(x(t)) = 256$ states and $\dim(u(t)) = \dim(y(t)) = 4$ inputs and outputs.
For the computation of the reduced order model, the empirical linear cross Gramian with an impulse input $u_i(t) = \delta(t)$ and a zero initial state $x_{0,i} = 0$ is used,
while for the trajectory simulation the default SSP32 (\cref{sec:rk}) integrator is utilized.

To quantify the quality of the resulting reduced order models,
the error between the original full order model output and the reduced order model's output is compared in the time-domain $\lnorm_2$-norm,
\begin{align*}
 \|y-\tilde{y}\|_{\lnorm_2} = \sqrt{\int_0^\infty \|y(t)-\tilde{y}(t)\|_2^2 \D t}.
\end{align*}
Also, the balanced truncation upper bound is assessed \cite{morAnt05}:
\begin{align*}
 \|y-\tilde{y}\|_{\lnorm_2} \leq 2 \|u\|_{\lnorm_2} \sum_{k=n+1}^N \sigma_k,
\end{align*}
for a reduced model of order $n$.
Instead of impulse input, zero-centered, unit-variance Gaussian noise is used as input time series for the evaluation.

\begin{figure}[t]\centering
 \includegraphics[width=.95\textwidth]{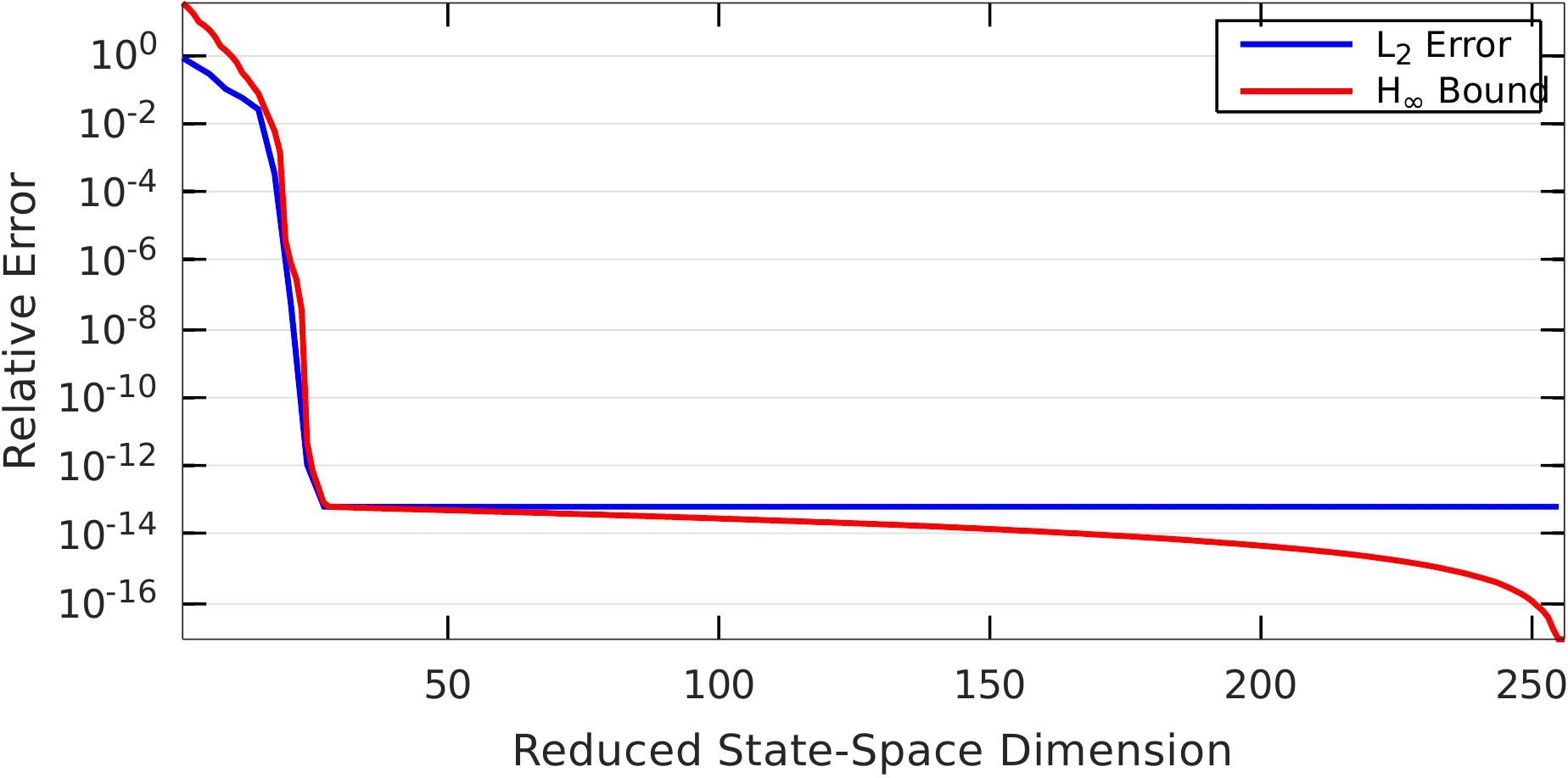}
 \caption{Model reduction error for the linear test system, see \cref{sc:numex1}.}
 \label{fig:lin}
\end{figure}

\cref{fig:lin} shows the reduced order model's relative $\lnorm_2$-norm model reduction error,
as well as the upper bound for increasing reduced orders.
The error evolves correctly tightly below the bound,
until the numerical accuracy\footnote{Double-precision floating point arithmetic.} is reached.

\subsection{Hyperbolic Evaluation}\label{sc:numex2}
The second numerical example is given by a one-dimensional transport equation.
An input-output system is constructed by selecting the left boundary as the input and the right boundary as the output:
\begin{align*}
 \frac{\partial}{\partial t} z(x,t) &= -\theta \frac{\partial}{\partial x}, \quad x \in [0,1], \\
 z(0,t) &= u(t), \\
 z(x,0) &= 0, \\
 y(t) &= z(1,t),
\end{align*}
while the transport velocity $\theta \in [1.0, 1.5]$ is treated as a parameter.
This partial differential equation system is spatially discretized by a first-order finite-difference upwind scheme,
yielding a SISO ordinary differential equation system:
\begin{align*}
 \dot{x}(t) &= A(\theta) x(t) + b u(t), \\
       y(t) &= c x(t),
\end{align*}
with $A(\theta) = \theta A$.
For this example, a spatial resolution of $\dim(x(t)) = 256$ is chosen,
hence \mbox{$A \in \R^{256 \times 256}$}, $b \in \R^{256}$ and $c \in \R^{256}$.
Since the system matrix is non-normal, using techniques such as POD may lead to unstable reduced order models.
Thus, here a cross-Gramian-based balancing technique is used guaranteeing stability of the reduced model.
For training, impulse responses for the extremal velocities are used;
for testing, a Gauss bell input is utilized over ten uniformly random velocities, both utilizing the default integrator.
The reduced order model quality is evaluated by the parametric norm \cite{morHim17}:
\begin{align*}
 \|y(\theta) - \tilde{y}(\theta)\|_{\lnorm_2 \otimes \lnorm_2} &= \sqrt{\int_\Theta \|y(\theta)-\tilde{y}(\theta_r)\|_{\lnorm_2}^2 \D \theta}.
\end{align*}

\begin{figure}[t]\centering
 \includegraphics[width=.95\textwidth]{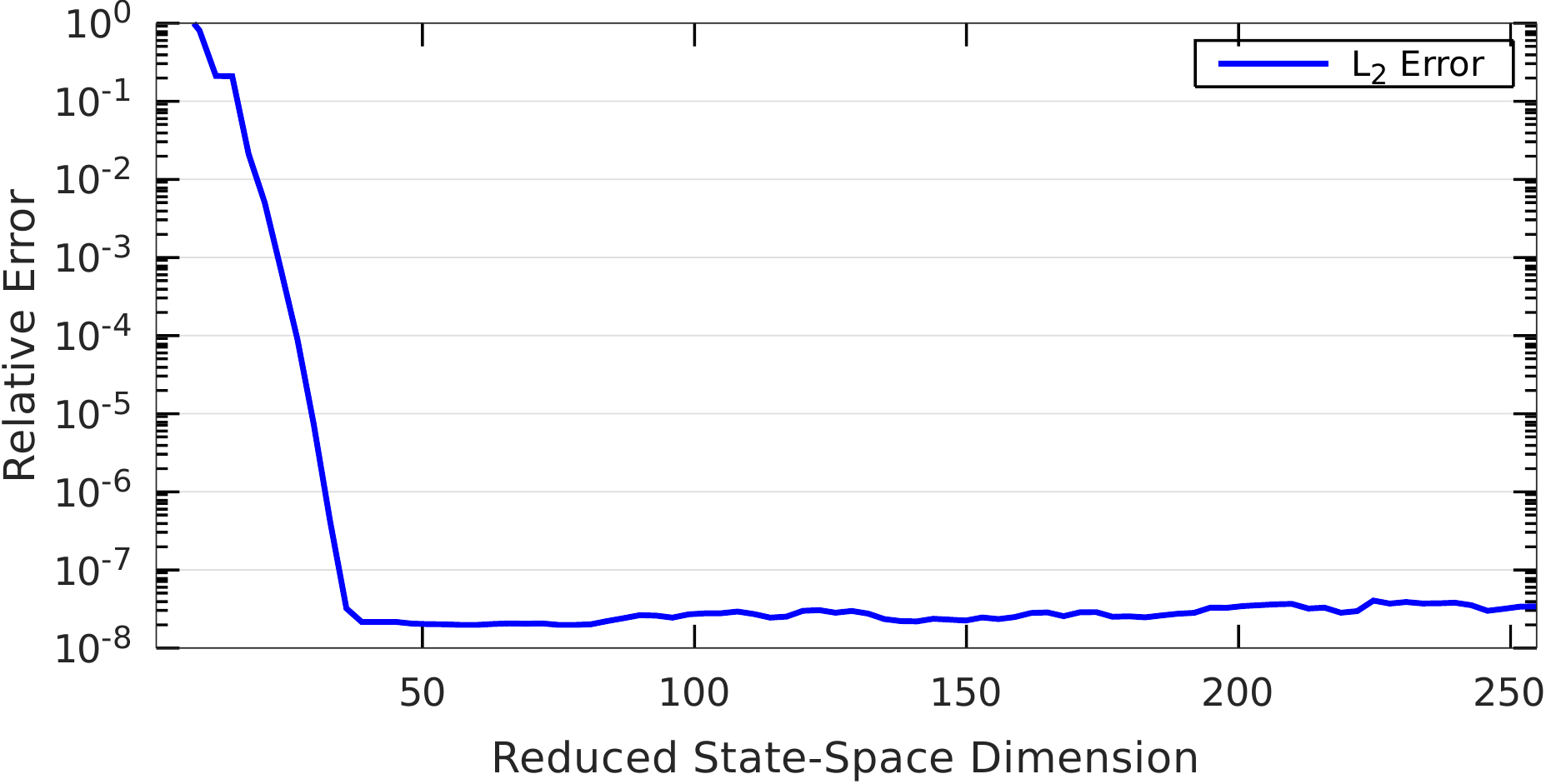}
 \caption{Model reduction error for the hyperbolic test system, see \cref{sc:numex2}.}
 \label{fig:hyp}
\end{figure}

Even though the system is hyperbolic, a steep decay in error is obtained.
Yet, due to the hyperbolicity and the parameter-dependence of the reduced order model, a lower overall numerical accuracy ($\approx 10^{-7}$) is achieved.

\subsection{Nonlinear Validation}\label{sc:numex3}
The third example involves a parametric nonlinear system,
based on the hyperbolic network model \cite{QuaZC01},
\begin{align*}
 \dot{x}(t) &= A \tanh(K(\theta)x(t)) + B u(t), \\
       y(t) &= C x(t).
\end{align*}
The structure of this system is similar to the linear system model \eqref{eq:linsys},
yet the vector-field includes a hyperbolic tangent nonlinearity,
in which the parametrized activation is described by a diagonal gain matrix $K(\theta)$, $K_{ii} = \theta_i$.
A negative Lehmer matrix\footnote{A Lehmer matrix is defined as $A_{ij} := \min(i,j) / \max(i,j)$, and is positive definite.} is selected as system matrix $A \in \R^{256 \times 256}$,
a vector of sequential cosine evaluations as input matrix $B \in \R^{256 \times 1}$, a binary matrix $C^{4 \times 256}$ as output matrix,
and parameters $\theta \in \R^{256}$ constrained to the interval $\theta_i \in [\frac{1}{2},1]$.
For this system a combined state and parameter reduction is demonstrated.
To this end an empirical non-symmetric joint Gramian is computed,
using again an impulse input $u(t) = \delta(t)$, a zero initial state $x_{0,i} = 0$ and the default integrator.
The reduced order model quality is evaluated for the same input and initial state by the joint state and parameter norm $\|y(\theta) - \tilde{y}(\theta_r)\|_{\lnorm_2 \otimes \lnorm_2}$, with respect to the reduced parameters for ten uniformly random samples from the admissible parameter-space.

\begin{figure}[t]\centering
 \includegraphics[width=.95\textwidth]{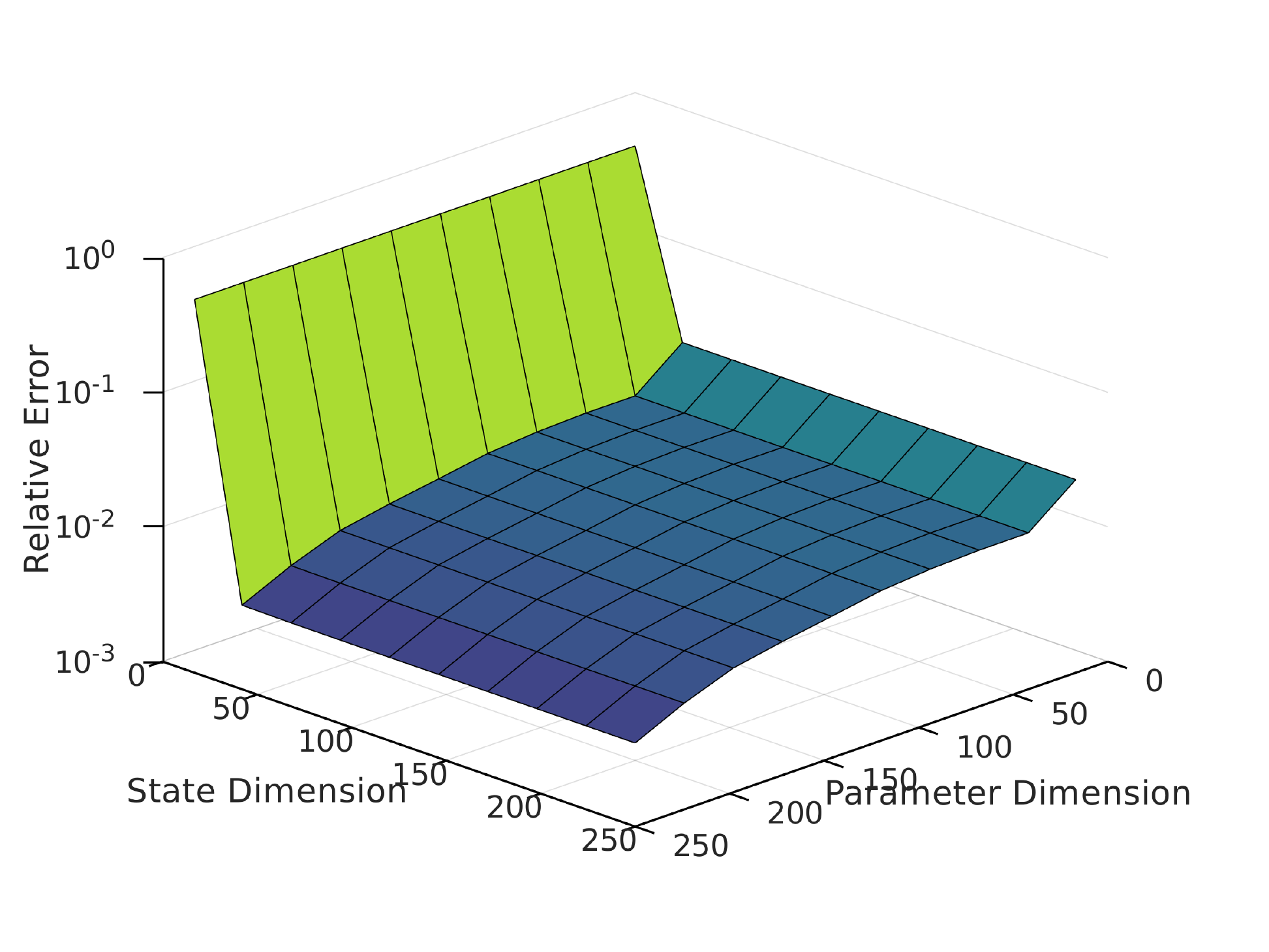}
 \caption{Model reduction error for the nonlinear test system, see \cref{sc:numex3}.}
 \label{fig:non}
\end{figure}

\cref{fig:non} depicts the $\lnorm_2 \otimes \lnorm_2$-norm model reduction error for increasing state- and parameter-space dimensions.
The combined reduction errors decays for both, reduced state-space and reduced parameter-space, yet faster for the state-space.
As for the parametric model, the numerical accuracy is reduced due to the combined reduction.


\section{Concluding Remark}\label{sec:sum}
Empirical Gramians are a universal tool for nonlinear system and control theoretic applications with a simple, data-driven construction.
The empirical Gramian framework - \emgr implements empirical Gramian computation for system input-output coherence and parameter identifiability evaluation.
Possible future extensions of \emgr may include Koopman Gramians \cite{morYeuLH17}, empirical Riccati covariance matrices \cite{morChoS17}, or empirical differential balancing \cite{morKawS17}.
Finally, further examples and applications can be found at the \emgr project website: \mbox{\url{http://gramian.de}}.

\section*{Code Availability}
The source code of the presented numerical examples can be obtained from:
\begin{center}
\url{http://runmycode.org/companion/view/2077}
\end{center}
and is authored by: \textsc{Christian Himpe}.

\section*{Acknowldgements}
Supported by the German Federal Ministry for Economic Affairs and Energy,
in the joint project: 
``\textbf{MathEnergy} -- Mathematical Key Technologies for Evolving Energy Grids'',
sub-project: Model Order Reduction (Grant number: 0324019\textbf{B}).

\bibliographystyle{plainurl}
\bibliography{mor,csc,software}



\end{document}